\documentclass[letterpaper, 10 pt, conference]{ieeeconf}  

\IEEEoverridecommandlockouts                              

\pdfoutput=1 

\usepackage{graphics} 
\usepackage{amsmath} 
\usepackage{amssymb}  
 \usepackage{mathptmx}       
\usepackage{helvet}         
\usepackage{courier}        
\usepackage{type1cm}        
\usepackage{comment}                          
\usepackage{makeidx}         
\usepackage{graphicx}        
\usepackage{multicol}        
\usepackage[bottom]{footmisc}
\usepackage{amsmath,amsfonts,amssymb}
\usepackage{amsmath,nccmath}
\usepackage{bbm}

\usepackage{cleveref}

\newtheorem{assumption}{Assumption}

\newtheorem{problem}{Problem}

\newtheorem{lemma}{Lemma}

\newtheorem{Theorem}{Theorem}

\title{\LARGE \bf
An Algorithm for Multi-Objective Multi-Agent Optimization
}

\author{Maude J. Blondin$^1$ and Matthew Hale$^{2}$
\thanks{*Maude J. Blondin was supported by Fonds de recherche Nature et technologies postdoctoral fellowship. Matthew Hale was supported by ONR grant N00014-19-1-2543, the AFOSR Center of Excellence in Assured Autonomy in Contested Environments, and a task order contract with AFRL Munitions Directorate at Eglin AFB.}
\thanks{$^{1}$Maude J. Blondin is a postdoctoral researcher with the Departement of Mechanical and Aerospace engineering, University of Florida, Gainesville, FL.
        {\tt\small maude.blondin@ufl.edu}}%
\thanks{$^{2}$ Matthew Hale is an assistant professor with the Departement of Mechanical and Aerospace engineering, University of Florida, Gainesville, FL.
        {\tt\small matthewhale@ufl.edu}}%
}

\begin{document}

\maketitle
\thispagestyle{empty}
\pagestyle{empty}

\begin{abstract}
Multi-agent optimization problems with many objective functions have drawn much interest over the past two decades. Many works on the subject minimize the sum of objective functions, which implicitly carries a decision about the problem formulation. Indeed, it represents a special case of a multi-objective problem, in which all objectives are prioritized equally. To the best of our knowledge, multi-objective optimization applied to multi-agent systems remains largely unexplored. Therefore, we propose a distributed algorithm that allows the exploration of Pareto optimal solutions for the non-homogeneously weighted sum of objective functions. In the problems we consider, each agent has one objective function to minimize based on a gradient method. Agents update their decision variables by exchanging information with other agents in the network. Information exchanges are weighted by each agent's individual weights that encode the extent to which they prioritize other agents' objectives. This paper provides a proof of convergence, performance bounds, and explicit limits for the results of agents' computations. Simulation results with different sizes of networks demonstrate the efficiency of the proposed approach and how the choice of weights impacts the agents' final result.  
\end{abstract}

\section{INTRODUCTION}

Multi-agent systems have gained significant interest in the past two decades \cite{Qin16}-\cite{Wan16}. Technological advances have enabled the deployment of multi-agent networks to many engineering applications from commercial to military uses \cite{Cio04}-\cite{Mca07}. In multi-agent optimization problems with many objectives, the most common approach is optimizing the sum of these functions \cite{Ned09}-\cite{Ned01}. This can be a natural approach, but computing the sum implicitly carries a decision about the problem formulation. In particular, the sum of objective functions represents a special case of a multi-objective problem, in which all objectives are prioritized equally. 

However, it is not difficult to envision cases in which there are objectives with different importance. For example, teams of robots may want to explore different regions of an area, or agents may have different priorities in trajectory planning when minimizing both energy consumption and travel time. In the centralized case, there has emerged a large body of work on multi-objective optimization to solve problems of this kind, such as the Tchebycheff method \cite{Mie99} and the Fandel method \cite{Sia04}. These and other algorithms are surveyed in \cite{Mie99}\cite{Sia04}. The goal of such algorithms is to explore a continuum of Pareto optimal solutions through different prioritizations of the objective functions of the problem. In light of these existing techniques, it is apparent that the sum of objectives is a special case of a broader class of multi-objective problems, where the solution to such problems represents an individual element of the Pareto Front. Further exploring this front can provide other solutions that are optimal in different senses, as well as provide a wider range of operating conditions for systems based on agents' needs. To the best of our knowledge, such techniques remain largely unexplored in a multi-agent context. 

In this paper, we propose a distributed algorithm that allows the exploration of the Pareto Front for multi-agent multi-objective problems. In particular, a team of $n$ agents optimizes the weighted sum of cost functions $f(x)=\sum_{i=1}^n w_if_i(x)$, where agent $i$ is tasked with minimizing $f_i$. When the proposed algorithm begins, each agent has an initial vector of priorities encoded as weights and an initial vector of decision variables. The proposed algorithm then has two update steps at each iteration: \textit{i)} agent $i$ updates its vector of priorities using those received from other agents in the network and, \textit{ii)} agent $i$ updates its decision vector with a gradient descent step and averaging the iterates of its neighbors.

Our proposed algorithm belongs to a broad class of averaging-based distributed optimization algorithms e.g. \cite{Ned09}\cite{Ned10}-\cite{Duc11}. What distinguishes our approach from the existing literature is that our algorithm allows the exploration of the Pareto Front through priorities selected by each agent independently. These priorities are used by agents to prioritize information they receive from others. Contrary to a large body of related work, this setup provides a network-level weight matrix that is not doubly-stochastic. Several works use the double-stochasticity assumption in their model and provide convergence rates and proofs of convergence using the infinite product of doubly-stochastic matrices \cite{Ned09}\cite{Ned10a}\cite{Ned10}\cite{Zha14}-\cite{Bia11}. Because this assumption does not hold here, we require a different approach that avoids it while accounting for the evolution of agents' weights over time.

In this paper, we provide a multi-agent multi-objective optimization algorithm that operates over a static, undirected graph with time-varying weights. We characterize this algorithm through convergence rates we derive. Through numerical simulations, we show that agents' initial priority vectors directly influence the final results of their computations. These priorities  are also shown to affect the algorithm's convergence rate. The major outcome of the proposed algorithm is that a multi-agent system can consider a wider range of solutions that are optimal in different senses, which paves the way for new applications with many objectives of different relative importance to different agents.

The rest of the paper is organized as follows. Section \ref{GRAPH THEORY AND MULTI-AGENT SYSTEM} presents background on graph theory and multi-agent interactions. The multi-agent optimization model and the proposed distributed optimization algorithm are provided in Section \ref{MULTI-AGENT MODEL AND ITS OPTIMIZATION}. Section \ref{Convergence of the weight vector update and the transition matrix} provides proofs of convergence of the priority-update rule, and the limit of agents' optimization updates. Section \ref{GRADIENT CONVERGENCE} presents the convergence rate for the overall multi-objective optimization algorithm. Section \ref{NUMERICAL RESULTS} presents numerical results, and Section \ref{CONCLUSIONS} concludes the paper.

\section{GRAPH THEORY AND MULTI-AGENT INTERACTIONS}
\label{GRAPH THEORY AND MULTI-AGENT SYSTEM}

In this paper, agents' interactions are represented by a connected and undirected graph $G = (V,E)$, where $V= [n]:=\{1,2, \ldots, n\}$ is the set of agents and $E  \: \subset V \times  V$ is the set of edges. An edge exists between agent $i$ and $j$, i.e., $(i,j) \in E$, if agent $i$ communicates with agent $j$. By convention, $(i,i) \notin E $ for all $i$. The adjacency matrix of the graph $G$ is represented with $H(G)$.  The degree of agent $i$ is the total number of agents that agent $i$ communicates with, denoted $deg(i)$. The degree matrix, $\Delta(G)$,  is a  diagonal $n \times n$ matrix, with $deg(i)$ on its diagonal for $i=1, \ldots, n$. The maximum vertex degree of $\Delta(G)$ is $ \Delta_{max} = \underset{i \in [n] }{max} \;\: deg(i)$.

\noindent Since $G$ is an undirected graph without self-loops, $H(G)$ is symmetric with zeros on its main diagonal. The Laplacian matrix associated with G is also symmetric, and is defined as 

\begin{equation}
 L(G) = \Delta(G) -H(G).   
\end{equation}

In this paper, we consider a static graph $G$, and, because $G$ is unambiguous, we will simply write its Laplacian as $L$.

\section{MULTI-AGENT OPTIMIZATION MODEL}
\label{MULTI-AGENT MODEL AND ITS OPTIMIZATION}

In this section, we formally define the class of problems to be solved. Then we propose a  multi-agent update law for solving them. 

\subsection{Problem Formulation}

In this paper, we consider problems in which agents minimize a prioritized sum of objective functions. Agent $i$ is responsible for minimizing the function $f_i$ about which we make the following assumption. 

\begin{assumption} For all $i \in \{1, \ldots, n\}$, the function $f_i : \mathbb{R}^n \rightarrow \mathbb{R} $ is continuously differentiable and convex.  
\label{ConvexAssump}
\end{assumption} \hfill $\triangle$

\noindent We next consider the following  optimization problem. 
\begin{problem} \text{Given convex  functions} $\{f_i\}_{i \in \{1, \ldots, n\}}$ satisfying \textit{Assumption} \ref{ConvexAssump},
\begin{equation}
\underset{x}{\text{minimize}} \sum_{i=1}^n w_if_i(x),
\end{equation}
\end{problem}

\noindent where $x$ is the vector of decision variables, $w_i$ is a priority assigned to $f_i$, $\sum_{i=1}^n w_i = 1$, and $ 0 < w_i < 1 $ for all $i$.  \hfill $\diamond$ 
\\

For a centralized problem, one could fix priorities $\{w_i\}_{i \in \{1, \ldots, n\}}$ and solve Problem 1 using a typical convex optimization method. However, in the decentralized case agents may choose different priorities. Formally, agent $i$ may choose $\{w_l^i\}_{l \in \{1, \ldots, n\}}$ while agent $j$ chooses $\{w_l^j\}_{l \in \{1, \ldots, n\}}$, with $w_l^i \neq w_l^j$ for all $l$. When this occurs, these priorities have the benefit of giving each agent the flexibility to have preferences. For example, agents generating a trajectory may wish to minimize fuel usage and travel time, and each agent can prioritize these two objectives according to their own needs. If agents select the same priorities, then they can optimize using standard techniques and reach a common solution. If their priorities are different, then agents are essentially solving different problems because they minimize different overall objective functions. As a result, reaching a common solution requires not only crafting an optimization algorithm, but also driving agents' priorities to a common value. 

Changing agents' priorities from their initial values means that no single agent's preferences are obeyed exactly. However, the net change across all agents can be done in a fair way. One such way is to drive all agent's priorities to their average value. While one could envision first computing the average priorities and then optimizing, this approach is undesirable because it requires solving two separate problems sequentially. Instead, we devise an update law that drives agents to a common solution by interlacing optimization steps with priority averaging steps. Also, this interlacing optimization allows agents to continuously modify their priorities based upon the task at hand.

\subsection{Proposed Update Law}

Each agent optimizes its assigned function using gradient descent. We choose a gradient-based algorithm due to its ease of implementation in a multi-agent context and its inherent robustness to disagreement among agents. Each agent $i$ has its own initial decision vector, denoted $x^i(0)= (x^i_1(0), \dots, x^i_m(0)) $, and its own initial vector of priorities, $w^i(0) = (w^i_1(0), \dots, w^i_n(0)) $, where $m$ and $n$ are the numbers of decision variables and agents, respectively. Thus, each agent assigns a priority to all agents, even if there is no communication between them. This provides each agent with a way to influence all final priorities, and, as will be shown below, influence the final results agents attain. \\

At iteration $k$, agent $i$ updates its priority vector $w^i$ and decision vector $x^i$ using 

\begin{equation}
    w^i(k+1) =  w^i(k) + c \sum_{j=1}^n h_j^i (w^j(k) - w^i(k))
  \label{w_i(k+1)}
\end{equation}

 \begin{equation}
x^i(k+1) = \sum_{j=1}^n a^i_j(k)x^j(k) - \alpha^i(k)d^i(k),
\label{x^i(k+1)}
\end{equation}

\noindent where $0 \:< \: c \: < \: 1/\Delta_{max}$, $h_j^i(k)$ is the j$^{th}$  i$^{th}$ entry of the $H(G)$,  $a_j^i(k)$ is the weight that agent $i$ assigns to the data provided by agent $j$ at iteration $k$, $\alpha^i$ is the step size of agent $i$, and $d^i$ is the gradient vector of agent $i$ at $x^i(k)$. Formally, $d^i(k) = \nabla f_i(x^i(k))$. We use $A(k)$ to denote the matrix used to update the decision variables. Each row of $A(k)$, $a_i(k)$ for $i=\{1, \ldots, n\}$, is the vector of weights that agent $i$ assigns to the other agents. 
As mentioned above, the literature contains a large body of work with doubly-stochastic $A(k)$ matrices. Indeed, some published rules ensure the double-stochasticity of $A(k)$, such as Metropolis-based weights \cite{Xia07} and the equal-neighbor model  \cite{Ols11}\cite{Blo05}. However, contrary to those existing works, the $A(k)$ matrix in our approach is not doubly-stochastic, but only row stochastic. This occurs because an agent can ensure that its own weights sum to 1, though different agents' weights for a particular objective need not to sum to 1. This implies that $A(k)$'s column sums need not to equal 1.

In line with the multi-objective concept, our algorithm uses the priority vectors, $w^i$ for $i=\{1, \ldots, n\}$, to quantify the importance of information received to update $x^i$ for $i=\{1, \ldots, n\}$.  Although agent $i$ assigns a priority to all other agents in $w^i$, it does not communicate with all other agents because agents' communication topology only needs to be connected. The question is then how we should compute $a^i(k)$ using $w^i(k)$, i.e., how agents should use preferences to weight information from the agents they communicate with. The proposed rule allocates to agent $i$ the priorities of agents $j$ if $(i,j) \notin E$. This allows agent $i$ to keep the relative importance given to agent $j$ for $(i,j) \in E$. The $A$ matrix is then computed as follows: 

 \begin{equation}
A(k) = Z \circ W(k) + (W(k)\circ \Tilde{H})J\circ I, 
\label{Acomputation}
\end{equation}

\noindent where $I$ is the identity matrix, $J$ is an all-ones matrix, $Z = H + I$, $\Tilde{H}=J-Z$ and $
W = \begin{bmatrix} w^1_1(k) &  w^1_2(k) & \ldots & w^1_n(k)\\
w^2_1(k) & w^2_2(k) & \ldots & w^2_n(k) \\
\vdots & \vdots & \ddots &\vdots  \\
w^n_1(k) & w^n_2(k) & \ldots & w^n_n(k)  
\end{bmatrix}. \\
$ \\ 

\noindent In \eqref{Acomputation}, $Z \circ W(k)$ computes the Hadamard product between $Z$ and $W(k)$. By doing this, the resulting matrix contains $w_j^i(k)$ for $(i,j) \in E$, $w_i^i(k)$ for all $i$, and the remaining terms are set to zero. In other words, if agent $i$ does not communicate with agent $j$, a zero is assigned to that agent. To analyse the second term of \eqref{Acomputation}, we will first explain the meaning of $\Tilde{H}$. $\Tilde{H}$ is the complement matrix of $Z$, i.e, 1 is assigned to $[\Tilde{H}]^j_i$ if $(i,j) \notin E$, and the remaining values are 0. Therefore, $(W(k)\circ \Tilde{H})J\circ I$ creates a diagonal matrix, where the diagonal terms are the sum of each row $[W(k) \circ \Tilde{H}]_i$ for $i=\{1, \ldots, n\}$. The first term summed to the second term  in \eqref{Acomputation} means that agent $i$ assigns to itself the weights $w_j^i$ if $(i,j) \notin E$ and assigns a zero value to the entries of the \textit{i}-th row and \textit{j}-th colum for $(i,j) \notin E$. The agent network topology is static, but the the weights are time-varying. The next assumption pertains to the weights of the $A$ matrix and the communication between agents.  \\

\begin{assumption} There exists a scalar  $0 < \eta_A <1 $ such that
\begin{itemize}
    \item  $a^i_i(k) \geq \eta_A$ for all $k \geq 0$ and all $i$.
    \item $a^i_j(k) \geq \eta_A $ for all $k \geq 0$ and all $(i,j) \in E$.
    \item $a^i_j(k) = 0$ for all $k$ if $(i,j) \notin E$.
    \item $A$ is row stochastic, i.e. the row sum of the $A$ matrix equals to 1.
\end{itemize}
\label{Assumption_etaA}
\end{assumption} \hfill $\triangle$

It is under \textit{Assumptions} 1 and 2 that we will perform our convergence analysis.
To simplify the forthcoming developments, the priority vector update \eqref{w_i(k+1)} is reformulated at the network level as  \cite{Olf07}

\begin{equation}
    W(k+1) =  PW(k),
  \label{ConsensusEq2}
\end{equation}

\noindent where $P = I - cL$, with $c \in (0, \: 1/\Delta_{max})$. The decision variable vector update \eqref{x^i(k+1)} is expressed as 

\begin{equation}
\begin{split}
    x^i(k+1) = & \sum_{j=1}^m [\Phi(k,s)]^j_i x^j(s) \\& - \sum_{r=s+1}^k \bigg( \sum_{j=1}^m [\Phi(k,r)]^j_i \alpha^j(r-1)d^j(r-1)\bigg) \\& - \alpha^i(k)d^i(k),
\end{split}
\label{x^i(k+1)_phiform}
\end{equation}

\noindent where the transition matrix $\Phi(k,s) = A(k)A(k-1), \ldots, A(s)$ \cite{Ned09}. We assume that $||\alpha^id^i(k)|| \leq L_1$ for all $i$, which simply means that the size of each agent's updates must be bounded. The next section analyzes the evolution of priorities and its impact upon agents' optimization updates. \\

\section{Convergence of the weight vector update and the transition matrix}
\label{Convergence of the weight vector update and the transition matrix}

In this section, we establish the convergence of the priority vector update \eqref{w_i(k+1)}. We also derive an explicit form for the weight bound $\eta_A$, as well as the limit of the transition matrix $\Phi(k,s)$. Both will be used in showing convergence of the optimization update in \eqref{x^i(k+1)_phiform}.

\subsection{Convergence of the weight vector update}

We state the following well-known Lemma, which confirms that the priority update \eqref{w_i(k+1)} does indeed compute average priorities.

\begin{lemma}$\lim_{k \to \infty} w^i(k) = \overline{w} = \sum_{j=1}^n \dfrac{w^j(0)}{n}$ for $j=1, \ldots, n$.  At the network level, $W(k)=\overline{W}$,where $\overline{W} = \boldsymbol{1} \overline{w}$ and $\boldsymbol{1}$ is an all-ones vector. 
\label{Lemma_overline{w}}
\end{lemma}

\begin{proof}
See \cite{Olf07}\cite{Khi07}.
\end{proof}

\subsection{Lower bound on $\eta_A$}

This subsection establishes an explicit lower bound on $\eta_A$, defined in \textit{Assumption} \ref{Assumption_etaA} as a lower bound on all non-zero weights. Although \textit{Assumption} \ref{Assumption_etaA} is common, the constant $\eta_A$ is typically not known. Here we are able to leverage the priority-based update law to compute $\eta_A$. The following Lemma will be used in doing so. 

\begin{lemma}
In the priorities update \eqref{w_i(k+1)}, we have $\underset{j \in [n]}{\text{min}} \; \underset{i \in [n]}{\text{min}} \: w^j_i(k+1) \geq \underset{j \in [n]}{\text{min}} \; \underset{i \in [n]}{\text{min}} \: w^j_i(k) $ for all $i,j,k$.
\end{lemma}

\begin{proof}
Define $ \mu(k) := \underset{j \in [n]}{\text{min}} \; \underset{i \in [n]}{\text{min}} \: w^j_i(k)$. Then, $W(k+1) = PW(k)$ can be expressed as

\begin{equation}
\begin{split}
 \begin{bmatrix} 
w_{1}^1(k+1) &\ldots &w_{1}^n(k+1) \\
\vdots & \vdots & \vdots \\
w_{i}^1(k+1) &\ldots &w_{i}^n(k+1) \\
\vdots & \vdots & \vdots \\
w_{n}^1(k+1) &\ldots &w_{n}^n(k+1) 
\end{bmatrix} 
\quad = \\
\begin{bmatrix} 
p_{1}^1 &\ldots &p_{1}^n \\
\vdots & \vdots & \vdots \\
p_{i}^1 &\ldots &p_{i}^n \\
\vdots & \vdots & \vdots \\
p_{n}^1 &\ldots &p_{n}^n 
\end{bmatrix}
\quad
    \begin{bmatrix} 
\mu(k) + \delta_{1}^1(k) &\ldots &\mu(k) + \delta_{1}^n(k) \\
\vdots & \vdots & \vdots \\
\mu(k) + \delta_{i}^1(k) &\ldots & \mu(k) + \delta_{i}^n(k) \\
\vdots & \vdots & \vdots \\
\mu(k) + \delta_{n}^1(k) &\ldots & \mu(k) + \delta_{n}^n(k)
\end{bmatrix} \\
\end{split}
\end{equation}

\noindent where $\delta_{i}^j(k) = w_{i}^j(k)-\mu(k) \geq 0$ for $i,j = \{1, \ldots, n\} $. Then, we have

\begin{equation}
\begin{split}
    w_{i}^j(k+1) = & \sum_{m=1}^n p_{i}^m [\mu(k) + \delta_{m}^j(k)] = \sum_{m=1}^n p_{i}^m \mu(k) + \sum_{m=1}^n p_{i}^m\delta_{m}^j(k) \\& = \mu(k) \sum_{m=1}^n p_{i}^m  + \sum_{m=1}^n p_{i}^m\delta_{m}^j(k).
    \end{split}
\end{equation}

\noindent By definition, we know that $\sum_{m=1}^n p_{i}^m =1$. Therefore, we get 
\begin{equation}
    w_{i}^j(k+1) = \mu(k) + \sum_{m=1}^n p_{i}^m\delta_{m}^j(k).
\end{equation}

\noindent Since $\delta_{m}^j \geq 0 $ and $p_i^m \geq 0$ for $i,j,m= \{1, \ldots,n\}$,  $w_{i}^j(k+1) \geq \mu(k) =\underset{j \in [n]}{\text{min}} \; \underset{i \in [n]}{\text{min}} \: w^j_i(k)$ for $i,j = \{1, \ldots, n \}$ and all $k$.  This establishes that the minimum of $W(k)$ is non-decreasing and other agents cannot go below the previous minimum at the next time step. 
\end{proof}
We next use the non-decreasing property of priorities to explicitly bound $\eta_A$.
\begin{Theorem}
  For $\eta_A$ in $Assumption$ \ref{Assumption_etaA}, we have $\eta_A = \underset{j \in [n]}{\text{min}} \; \underset{i \in [n]}{\text{min}} \: w^j_i(0)$.
\end{Theorem}

\begin{proof} With Lemma 1  and with equation \eqref{Acomputation} defining $A(k)$, the smallest non-zero element of $A(k)$, denoted $\underset{i \in [n]}{min^+ \:}\underset{j \in [n]}{min^+} [A(k)]^j_i$, is at least $\underset{i \in [n]}{min \:} \underset{j \in [n]}{min} w_i^j(k)$. This directly implies that the lower bound can be set as $\eta_A = \underset{j \in [n]}{\text{min}} \; \underset{i \in [n]}{\text{min}} \: w^j_i(0)$.
\end{proof}

\subsection{Convergence behavior of $\Phi(k,s)$ and its limit $\phi$}

The next Lemma describes the convergence behavior of $\Phi(k,s)$. We state it here because we will use it below.

\begin{lemma}
Suppose \textit{Assumption} \ref{Assumption_etaA}  holds. The convergence of $\Phi(k,s)$ is geometric according to \\

\begin{equation}
    |[\Phi(k,s)]^j_i - \phi^j(s)| \leq 2 \dfrac{(1 + \eta_A^{-B_0})}{1 - \eta_A^{B_0}}(1 - \eta_A^{B_0})^{(k-s)/B_0},
\end{equation}

\noindent where $B_0 = (n-1)$ and $n$ is the number of agents.  
\end{lemma}

\begin{proof}
 See Lemma 3 and Lemma 4 in \cite{Ned09}.
\end{proof}

Of course, in establishing the limit of $\Phi(k,s)$, we are interested in the exact value of $\phi$. The limit of $\phi$ has been established in \cite{Ned09} for a doubly-stochastic $\Phi(k,s)$, where $\phi= \dfrac{1}{m} \cdot \boldsymbol{1}$. In contrast, our proposed algorithm does not have this doubly-stochastic property. Therefore, in this subsection, we use the structure provided by the priority update rule to establish the limit of $\phi$ for a row stochastic $\Phi(k,s)$. For purposes of analysis, we first express the limit of $A$, denoted $\overline{A}$, as
\begin{equation}
\begin{split}
&    \overline{A} = \overline{W} + C,
    \end{split}
\end{equation}

\noindent where $C = (\overline{W} \circ \Tilde{H})J\circ I - \overline{W}  \circ \Tilde{H}$. To facilitate the convergence analysis, we can express C as 

\begin{equation}
    C = Q + F,
    \label{C}
\end{equation}  
\noindent where $Q = (\overline{W} \circ \Tilde{H})J\circ I$ and $F = - \overline{W}  \circ \Tilde{H}$. Therefore, the term $Q$ is the same as the second term of \eqref{Acomputation}, i.e., $(W \circ \Tilde{H})J\circ I$ . Since we are interested in the convergence behavior of $\Phi(k,s)$, $\overline{W}$ is used in \eqref{C}. The matrix $F$ is simply the Hadamard product of $\overline{W}$ and $\Tilde{H}$. We note that the limit $\overline{A}$ exists because the priority matrix converges to $\overline{W}$. Below, we will study products containing $\overline{A}$, and the following Lemma will be used to show that certain terms vanish. 

\begin{lemma}
Let $\overline{A} = \overline{W} + C$ as above. Then, \\
\textit{a)} $\overline{W}C=0$. \\
\textit{b)} $C\overline{W} = 0$. \\
\textit{c)} $(\overline{W}+C)^{r} = \overline{W}^r + C^{r}$. \\
\textit{d)} $\lim_{k\to\infty} \overline{W}^{k} = \overline{W}$.  \\
\textit{e)} $\lim_{k\to\infty} C^k = C^{\infty}$ = 0. \\
\end{lemma}

\begin{proof}
\textit{a)}  We will use the diagonal property of $Q$ to establish the proof. Therefore, we will express $\overline{W}C$ as $\overline{W}C = \overline{W}(Q + F) = \overline{W}Q + \overline{W}F$. Based on analysis in Section \ref{MULTI-AGENT MODEL AND ITS OPTIMIZATION}, the matrix $Q$ is 

\begin{equation}
Q =
\begin{bmatrix} 
\sum_{k=1}^n [\overline{W}]_1^k [\Tilde{H}]_1^k & 0 &\ldots & 0 \\
0 & \sum_{k=1}^n [\overline{W}]_2^k [\Tilde{H}]_2^k &\ldots & 0 \\
\vdots & \vdots & \ddots &\vdots \\
0 & 0 & \vdots  & \sum_{k=1}^n [\overline{W}]_n^k [\Tilde{H}]_n^k
\end{bmatrix}.
\end{equation}

\noindent Therefore, the entry of the $l$-th row and the $k$-th column of the product $\overline{W}Q$ is expressed as $[\overline{W}]_l^k\sum_{m=1}^n [\overline{W}]_k^m [\Tilde{H}]_k^m $.  The entry of the $l$-th row and $k$-th column of the product $\overline{W}F$ is expressed as $-\sum_{m=1}^n[\overline{W}]_l^m[\overline{W}]_m^k[\Tilde{H}]_m^k$ Thus, we can express $[\overline{W}C]_{l}^k$ as follows:

\begin{align*}
 [\overline{W}C]_{l}^k = [\overline{W}]_l^k\sum_{m=1}^n [\overline{W}]_k^m [\Tilde{H}]_k^m - \sum_{m=1}^n[\overline{W}]_l^m[\overline{W}]_m^k[\Tilde{H}]_m^k. \\
  \end{align*}

 \noindent From Lemma \ref{Lemma_overline{w}}, we know that $\overline{W}$ has identical rows, i.e, $[\overline{W}]_{m} = [\overline{W}]_{l}$ for $m= \{1, \ldots, n \}$. Thus, the entry $[\overline{W}]_m^k = [\overline{W}]_l^k$ for $m= \{1, \ldots, n\} $. We can then replace $[\overline{W}]_m^k$ in the previous equation by $[\overline{W}]_l^k$. We obtain 
  
\begin{align*}
 [\overline{W}C]_{l}^k = [\overline{W}]_l^k\sum_{m=1}^n [\overline{W}]_k^m [\Tilde{H}]_k^m - \sum_{m=1}^n[\overline{W}]_l^m[\overline{W}]_l^k[\Tilde{H}]_m^k. \\
  \end{align*}
  
  \noindent Since $ [\overline{W}]_{l}^k$ is a scalar, we factor it out of the summation. Then, we have
  
\begin{align*}
 [\overline{W}C]_{l}^k = [\overline{W}]_l^k\sum_{m=1}^n [\overline{W}]_k^m [\Tilde{H}]_k^m - [\overline{W}]_l^k \sum_{m=1}^n[\overline{W}]_l^m[\Tilde{H}]_m^k. \\
  \end{align*} 
  
\noindent $\Tilde{H}$ is symmetric. Therefore, the previous equation can be written as 

\begin{align*}
 [\overline{W}C]_{l}^k = [\overline{W}]_l^k\sum_{m=1}^n [\overline{W}]_k^m [\Tilde{H}]_k^m - [\overline{W}]_l^k \sum_{m=1}^n[\overline{W}]_l^m[\Tilde{H}]_k^m. \\
  \end{align*} 
  
  \noindent As mentioned above, $\overline{W}$ has identical rows, which means $[\overline{W}]_k^m = [\overline{W}]_l^m$ for $m=\{1, \ldots, n\}$. Therefore,  we can substitute $[\overline{W}]_k^m$ by $[\overline{W}]_l^m$ in the previous equation and we obtain for all $l$ and $k$, 
  
\begin{align*}
 [\overline{W}C]_{l}^k = [\overline{W}]_l^k\sum_{m=1}^n [\overline{W}]_l^m [\Tilde{H}]_k^m - [\overline{W}]_l^k \sum_{m=1}^n[\overline{W}]_l^m[\Tilde{H}]_k^m = 0.\\
  \end{align*}

\noindent \textit{b)} Similarly to the previous proof, $[C\overline{W}]_{l}^k$ can be written as follows:

\begin{align*}
 [C\overline{W}]_{l}^k =  \sum_{m=1}^n [\overline{W}_{l}^m \Tilde{H}_{l}^m]\overline{W}_{l}^k - \sum_{k=1}^n [\overline{W}_{l}^m  \Tilde{H}_{l}^m ] \overline{W}_{m}^k.   \\
  \end{align*}
  
\noindent Since $\overline{W}$ has identical rows,  $\overline{W}_{l}^k = \overline{W}_{m}^k$ for  $m=\{1, \ldots, n \}$, we find that, \\

\begin{align*}
 [C\overline{W}]_{l}^k =  \sum_{m=1}^n [\overline{W}_{l}^m \Tilde{H}_{l}^m]\overline{W}_{l}^k - \sum_{k=1}^n [\overline{W}_{l}^m  \Tilde{H}_{l}^m ] \overline{W}_{l}^k  = 0. \\
  \end{align*}
  
\noindent \textit{c)}  We prove this statement by induction on $r$. \textit{Base case :}  $(\overline{W}+C)^{2} = \overline{W}^2 + \overline{W}C + C\overline{W} + C^{2} = \overline{W}^2 + C^{2}$ since $C\overline{W} = 0$ and $\overline{W}C = 0$. 
\noindent \textit{Inductive step :} Now, we assume  that $(C + \overline{W})^{r-1} = C^{r-1} + \overline{W}^{r-1}$. Then, since $C\overline{W} = 0$ and $\overline{W}C = 0$, we obtain : 

\begin{align*}
\begin{split}
(C + \overline{W})^{r} & = (C + \overline{W})(C + \overline{W})^{r-1} \\
& = (C + \overline{W})(C^{r-1} + \overline{W}^{r-1}) \\
& = C^r + C\overline{W}^{r-1} + \overline{W}C^{r-1} + \overline{W}^r \\
& = C^r + \overline{W}^r. \\
\end{split}
\end{align*}

\noindent \textit{d)} Recall that $\overline{W}$ is row stochastic with identical rows. Therefore, $\overline{W}^{k} = \overline{W}$ for $  \forall \: k \geq 1$. \\

\noindent \textit{e)} Recall that C is defined as

\begin{equation}
C = Q + F,  
\end{equation}

\noindent where $Q = (\overline{W} \circ \Tilde{H})J\circ I$  and $F = - \overline{W}  \circ \Tilde{H}$.

\noindent Then the limit $\lim_{k \rightarrow \infty} (Q + F)^{k}$ leads to infinite products of Q and F matrices. By definition, the Q matrix is a diagonal matrix with the sum of each row of F on its diagonal (see proof Lemma 4a). Recall that $\sum_{j=1}^n [\overline{W}]_i^j = 1 $ for $i = \{1, \ldots, n\}$. Since an agent communicates with at least one other agent, the product $\overline{W} \circ \Tilde{H}$ sets at least one positive element of each row equal to zero. Then the absolute row sums of F are less than 1. This means that all eigenvalues of Q are less than 1. Similarly, the absolute sum of each row of F is less than 1, and diagonal F elements are 0. Therefore, according to Gershgorin circle theorem, all eigenvalues of F are less than 1 in magnitude. Thus, $F^{\infty}$ = 0 and $Q^{\infty}$ = 0, as well as all infinite products of $Q$ and $F$. Therefore, $\lim_{k \rightarrow \infty} (Q + F)^{k} = C^k = 0$. \\
\end{proof}

We next establish the limiting behavior of $\Phi$, which we will use later to analyze the convergence of each $x^i$.
 
 \begin{Theorem}
   $\overline{\Phi(k,s)} = \lim_{k \geq s, s \to\infty} \Phi(k,s) =  \prod_{n=1}^{\infty} \overline{A} = \overline{W} = \boldsymbol{1}\overline{w}=\boldsymbol{1}\phi$.
 \end{Theorem}

\begin{proof}
Because $A(k) \rightarrow \overline{A}$, for all $\epsilon > 0$, there exists a $\overline{k}$ such that for all $s > \overline{k}$ and $k > s $

\begin{align*}
    \mid\mid A(k) - \overline{A} \mid\mid < \epsilon.
\end{align*}

\noindent Therefore, writing $A(\ell) = \overline{A} + \epsilon M_{\ell}$ for some matrix $M_{\ell}$, we have

\begin{align*}
\begin{split}
\Phi(k,s) &= A(k)A(k-1), \ldots, A(s) \\&  = (\overline{A} + \epsilon M_k)(\overline{A} + \epsilon M_{k-1})(\overline{A} + \epsilon M_{s})  \\
&= \overline{A}^{k-s+1} + \epsilon X + O(\epsilon^2) \cdot I
\end{split}
\end{align*}
\noindent for some matrix $X$.

\noindent Since $\epsilon$ was arbitrary, all $\epsilon$ terms can be made arbitrarily small. Therefore, we have

\begin{equation}
   \lim_{\substack{s \to\infty \\ k \geq s }} \Phi(k,s) = (\overline{W}+C)^{k-s+1} = (\overline{W}+C)^r,
\end{equation}

\noindent where $r=k-s+1$. \\

From Lemma 4c, 
\begin{align*}
\lim_{\substack{s \to\infty \\ k \geq s }} (\overline{W}+C)^{k} = \lim_{\substack{s \to\infty \\ k \geq s }} \overline{W}^k + \lim_{\substack{s \to\infty \\ k \geq s }} C^{k}.    
\end{align*}
From Lemma 4d and 4e, 
\begin{align*}
\lim_{\substack{s \to\infty \\ k \geq s }} \overline{W}^k + \lim_{\substack{s \to\infty \\ k \geq s }} C^{k} = \overline{W}.
\end{align*}
Thus, $ \overline{\Phi(k,s)} = \lim_{\substack{s \to\infty \\ k \geq s}} \Phi(k,s)=\overline{W}=\boldsymbol{1}\phi$, where $\phi= \sum_{j=1}^n \dfrac{w^j(0)}{n}$ for $j=1, \ldots, n$.

\end{proof}

\subsection{Bound for $\phi$}

This subsection provides an explicit bound for $\phi$. As with the previous Lemmas and Theorems, the following theorem establishes a bound for $\phi$.

\begin{Theorem} The convergence bound $\Phi(k,s)$ is $\underset{j \in [n]}{min \phi^j} \geq \eta_A$ for $k \geq s$ and $ s \rightarrow \infty$.
\end{Theorem}

\begin{proof}
Recall that $\overline{w} = \sum_{j=1}^n \dfrac{w^j(0)}{n}$ for $i= \{1, \ldots, n\}$, $\eta_A = \underset{j \in [n]}{\text{min}} \; \underset{i \in [n]}{\text{min}} \: w^j_i(0)$, and $\phi = \overline{w}$.  Thus $\underset{j \in [n]}{min \phi^j} \geq \eta_A$ for $k \geq s$ and $ s \rightarrow \infty$. 
\end{proof}

\section{GRADIENT CONVERGENCE}
\label{GRADIENT CONVERGENCE}

This section studies the convergence of the agents' information update with \eqref{x^i(k+1)_phiform} and provides a performance bound for the algorithm. Our study draws from the work of Nedic and Ozdaglar \cite{Ned09}\cite{Ned10a}, who obtained a bound of performance for a related algorithm. Their analysis is based on the assumption that the transition matrix $\Phi(k,s)$ is doubly stochastic and  finds the limit $\phi= \dfrac{1}{m} \cdot \boldsymbol{1}$. In contrast, our proposed algorithm considers non-doubly stochastic matrices, which requires new analysis. 

This section contains three parts in order to establish the performance bound \cite{Ned10a}: \textit{i)} obtaining the disagreement estimate $|| x^i(k)-y(k)||$ relating to an auxiliary sequence $\{y(k)\}$, \textit{ii)} computing the estimate of the objective function $f(\hat{y}(k))$, \textit{iii)} establishing the performance bound, i.e. an estimate of $f(\hat{x}^i(k))$ as a function of the iteration number $k$ by using the results of the two previous parts, i.e., $||x^i(k)-y(k)||$ and $f(\hat{y}(k))$.

The following Lemma presents the performance bound for the algorithm, i.e., a performance bound with $\hat{x}^i(k)$ as

\begin{equation}
    \hat{x\:}^i(k)=\dfrac{1}{k}\sum_{h=1}^k x^i(h),
    \label{x_hat}
\end{equation}

\noindent where $x^i(h)$ is defined by \eqref{x^i(k+1)}. 

\begin{lemma} Let Assumptions 1 and 2 hold and let the set $X^*$ of optimal solutions of the problem be nonempty. Suppose the gradient is bounded as previously defined. Then, we have for all $i$ and $k\geq1$,
    \begin{equation}
\begin{split}
    f(\hat{x\:}^i(k)) \leq & f(x^*) + \dfrac{m^2 L \Omega}{ k\beta(1-\beta^{1/B_0}) }C_1   + \dfrac{\alpha L^2C_2 }{2 \underset{1 \leq j \leq n}{\text{min}}(\phi^j) }   \\& + \dfrac{(dist(y(0),X^*+\alpha mL)^2}{ 2k \alpha \underset{1 \leq j \leq n}{\text{min}}(\phi^j) } + \dfrac{2 \alpha m L^2}{k}
    \end{split}
\end{equation}

\noindent where $C_1 = \bigg[1 + 2\underset{1\leq j \leq m}{max}||x^j(0)|| \bigg[ \dfrac{2}{\underset{1 \leq j \leq n}{ \: \text{min}}(\phi^j)} + 1 \bigg] \bigg], C_2 = 8 m \bigg(1 + \dfrac{m \Omega}{\beta(1-\beta^{1/B_0})} \bigg) +  || \phi||^2m, \Omega= 1 + \eta_A^{-B_0}$, and $\beta=1 - \eta_A^{B_0}$.
\end{lemma}

\begin{proof}
The proof follows \cite{Ned10a} with the modification of the limit $\phi$ defined as $\underset{j \in [n]}{min}\phi^j$.
\end{proof}




\begin{Theorem}
For $s \rightarrow \infty$ and $k > s$, we have, \\
$  \underset{k \rightarrow \infty}{\limsup} f(\hat{x\:}^i(k)) \leq f(x^*) + \dfrac{\alpha L^2C_2 }{2  (\underset{j \in [n]}{\text{min}} \; \underset{i \in [n]}{\text{min}} \: w^j_i(0))} $
\end{Theorem}

\begin{proof}
The theorem is obtained by combining Theorem 3, Lemmas 1 and 5.
\end{proof}

Agents' initial priorities influence the performance bound for the proposed algorithm as shown above. This is not surprising since the final weights are computed from agents' initial priorities. An extremely small initial priority would increase the limit superior of the performance bound, which indicates that very small initial weights can harm performance. This suggests that agents' priorities must be balanced with need for attaining a high-quality final result, and agents can take this into account when calibrating their preferences.

\section{NUMERICAL RESULTS}
\label{NUMERICAL RESULTS}

In this section, the proposed algorithm is run over two simulation scenarios with different size networks. The first simulation goal is to show graphically the influence of the initial priorities on the final results, i.e., the objective function values and final decision variables. This simulation is performed with the quadratic functions 

\begin{equation}
\begin{split}
    f_1(x) &= 2(x-15)^2+100 \\
    f_2(x) &= 5(x+275)^2+10,000,
    \end{split}
\end{equation}

\noindent where $f_1(x)$ is associated to agent 1 and $f_2(x)$ to agent 2. Therefore, the agents solve
\begin{equation}
 \underset{x}{\text{minimize}} \sum_{j=1}^2 w_jf_j(x).
\end{equation}

The gradient step size is $\alpha$ = 0.00002. Twenty simulations with different initial priorities (see Table \ref{RandomInitialWeigths}) were performed for the same initial states, which are $x_1 = 485$ and $x_2 = 200$. Agents exchange information $100,000$ times, i.e, the maximum iteration number $k$ is 100,000 in \eqref{w_i(k+1)} and \eqref{x^i(k+1)}. Table \ref{ResultsF2} presents the results obtained by the proposed algorithm. The first two columns present the convergence of the priorities vector, where $\overline{w}_1$ and $\overline{w}_2$ are the average of the initial priorities, which agents' priorities converge to. Also, we see that the proposed algorithm closely approaches the optimal value $x^*$, which is shown by how closely $\hat{x}$ approaches $x^*$. By comparing the twenty simulations, we notice the direct influence of priorities on the final agent results. Indeed, the agent state value, $\hat{x}$, which is defined by \eqref{x_hat} with $k=100,000$, increases as $w_1$ value increases while the $\sum_{j=1}^2 \overline{w_j}f(\hat{x})$ value decreases. Figure \ref{F2vsF1} displays this influence, where the Pareto Front of $f_2(x)$ as a function of $f_1(x)$ is presented. The first right point on Figure \ref{F2vsF1}  corresponds to the first entry of Tables \ref{RandomInitialWeigths} and \ref{ResultsF2} and so on up to the last left point. Therefore, the proposed algorithm allows the exploration of the Pareto Front by optimizing several objective functions.

\begin{table}[h]
\centering
\caption{Different initial priorities for scenario 1}
\label{RandomInitialWeigths}
\resizebox{0.25\textwidth}{!}{
\begin{tabular}{lllll}
\hline
	&	\multicolumn{2}{c}{Agent 1}			&	\multicolumn{2}{c}{Agent 2}				\\ \hline
	&	$w_1$	&	$w_2$	&	$w_1$	&	$w_2$	\\ \hline
1	&	0.134	&	0.866	&	0.022	&	0.978	\\
2	&	0.577	&	0.423	&	0.026	&	0.974	\\
3	&	0.139	&	0.861	&	0.476	&	0.524	\\
4	&	0.561	&	0.439	&	0.269	&	0.731	\\
5	&	0.560	&	0.440	&	0.301	&	0.699	\\
6	&	0.521	&	0.479	&	0.372	&	0.628	\\
7	&	0.433	&	0.567	&	0.471	&	0.529	\\
8	&	0.647	&	0.353	&	0.308	&	0.692	\\
9	&	0.287	&	0.713	&	0.801	&	0.199	\\
10	&	0.447	&	0.553	&	0.646	&	0.354	\\
11	&	0.362	&	0.638	&	0.788	&	0.212	\\
12	&	0.849	&	0.151	&	0.373	&	0.627	\\
13	&	0.749	&	0.251	&	0.504	&	0.496	\\
14	&	0.549	&	0.451	&	0.728	&	0.272	\\
15	&	0.780	&	0.220	&	0.669	&	0.331	\\
16	&	0.896	&	0.104	&	0.598	&	0.402	\\
17	&	0.716	&	0.284	&	0.839	&	0.161	\\
18	&	0.937	&	0.063	&	0.830	&	0.170	\\
19	&	0.884	&	0.116	&	0.944	&	0.056	\\
20	&	0.939	&	0.061	&	0.981	&	0.019	\\ \hline
\end{tabular}}
\end{table}

\begin{table}[h]
\centering
\caption{Results obtained by the proposed algorithm for scenario 1}
\label{ResultsF2}
\resizebox{0.45\textwidth}{!}{
\begin{tabular}{lllllll}
\hline
&$\overline{w_1}$	&	$\overline{w_2}$	&	$\boldsymbol{x^*}$	&	$\hat{x}$	&	$\sum_{j=1}^2 \overline{w_j}f_j(x^*)$	&	$\sum_{j=1}^2 \overline{w_j}f(\hat{x})$	\\ \hline
1&0.078	&	0.922	&	-265.57	&	-265.56	&	21,849	&	21,849	\\
2&0.301	&	0.699	&	-232.34	&	-232.33	&	50,241	&	50,241	\\
3&0.307	&	0.693	&	-231.32	&	-231.31	&	50,840	&	50,840	\\
4&0.415	&	0.585	&	-210.92	&	-210.91	&	60,258	&	60,258	\\
5&0.430	&	0.570	&	-207.71	&	-207.70	&	61,325	&	61,325	\\
6&0.447	&	0.553	&	-204.20	&	-204.19	&	62,375	&	62,375	\\
7&0.452	&	0.548	&	-203.07	&	-203.06	&	62,688	&	62,688	\\
8&0.477	&	0.523	&	-197.42	&	-197.41	&	64,077	&	64,077	\\
9&0.544	&	0.456	&	-181.39	&	-181.38	&	66,550	&	66,550	\\
10&0.546	&	0.454	&	-180.70	&	-180.69	&	66,612	&	66,612	\\
11&0.575	&	0.425	&	-173.09	&	-173.08	&	67,064	&	67,064	\\
12&0.611	&	0.389	&	-163.16	&	-163.15	&	67,069	&	67,069	\\
13&0.626	&	0.374	&	-158.57	&	-158.56	&	66,863	&	66,863	\\
14&0.639	&	0.361	&	-154.86	&	-154.85	&	66,606	&	66,606	\\
15&0.724	&	0.276	&	-126.37	&	-126.36	&	62,224	&	62,224	\\
16&0.747	&	0.253	&	-118.03	&	-118.03	&	60,231	&	60,231	\\
17&0.777	&	0.223	&	-106.02	&	-106.01	&	56,865	&	56,865	\\
18&0.883	&	0.117	&	-56.99	&	-56.96	&	38,136	&	38,136	\\
19&0.914	&	0.086	&	-40.30	&	-40.25	&	30,263	&	30,263	\\
20&0.960	&	0.040	&	-12.26	&	-12.18	&	15,674	&	15,674	\\ \hline
\end{tabular}}
\end{table}

\begin{figure}[h]
   \includegraphics[width=0.5\textwidth,height=0.35\textwidth]{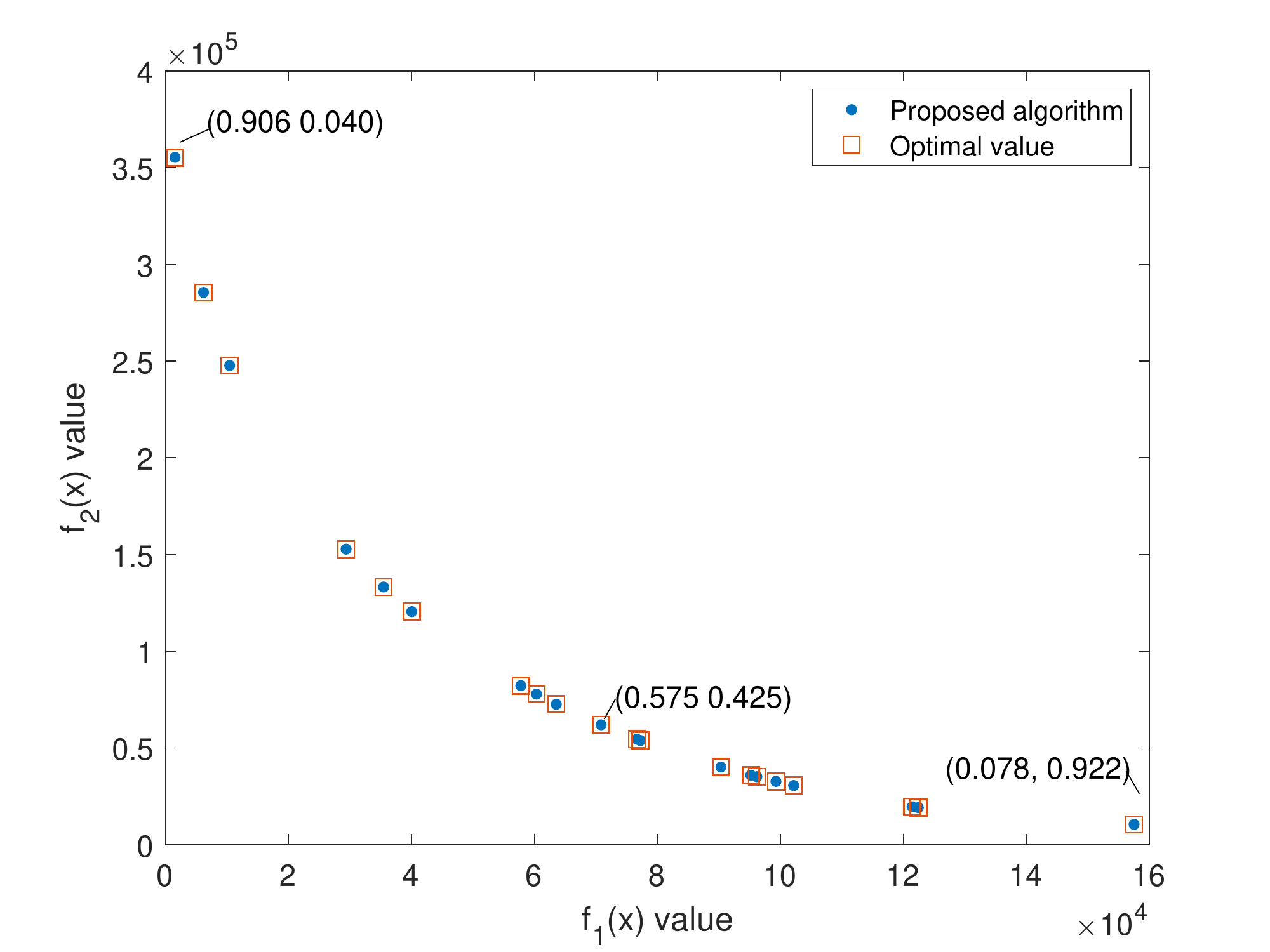}
   \caption{$f_2(x)$ in function of $f_1(x)$.}
    \label{F2vsF1}
  \end{figure}

A second simulation scenario was run consisting 20 agents that minimize the prioritized sum of the following objective functions:
 
 \begin{align*}
f_1(x) &= x_1^2+x_1x_2+x_2^2 &   f_{11}(x)& = 2x_1^2               \\
f_2(x) &= 5(x_1^2+x_1x_2+x_2^2)         &  f_{12}(x) &= x_1^2   \\
f_3(x) &= 10x_1+15x_2+20x_3   &  f_{13}(x) &= 5x_1+150   \\
f_4(x) &= \sum_{i=1}^{10} x_i^2 & f_{14}(x) &= \sum_{i=1}^6 x_i \\
f_5(x) &= e^{x_1} & f_{15}(x) &= 10(x_1+25)^2 \\
f_6(x) &= 3(x_1+17)^2+150 & f_{16}(x) &= e^{2x_1}+e^{3x_2}+e^{3x_3}+e^{3x_4} \\
 f_7(x) & = 30(x_1+3)^2+30 & f_{17}(x) &= \sum_{i=1}^6 x_i^2 \\
 f_8(x) &= 7(x_1-10)^2+10 & f_{18}(x) &= 15(x_1-15)^2-100 \\
 f_9(x)& = x_1^2+x_2^2 & f_{19}(x) &= 2(x_1+x_1x_2+x_2^2) \\
  f_{10}(x) &= x_1+x_2+x_3 & f_{20}(x) &= 100e^{x_1}.
\end{align*}

Each agent has a function assigned to it and starts with different priorities and state vectors. Figure \ref{20agents} shows the  network topology. 

  \begin{figure}[h!]
   \includegraphics[width=0.42\textwidth,height=0.35\textwidth]{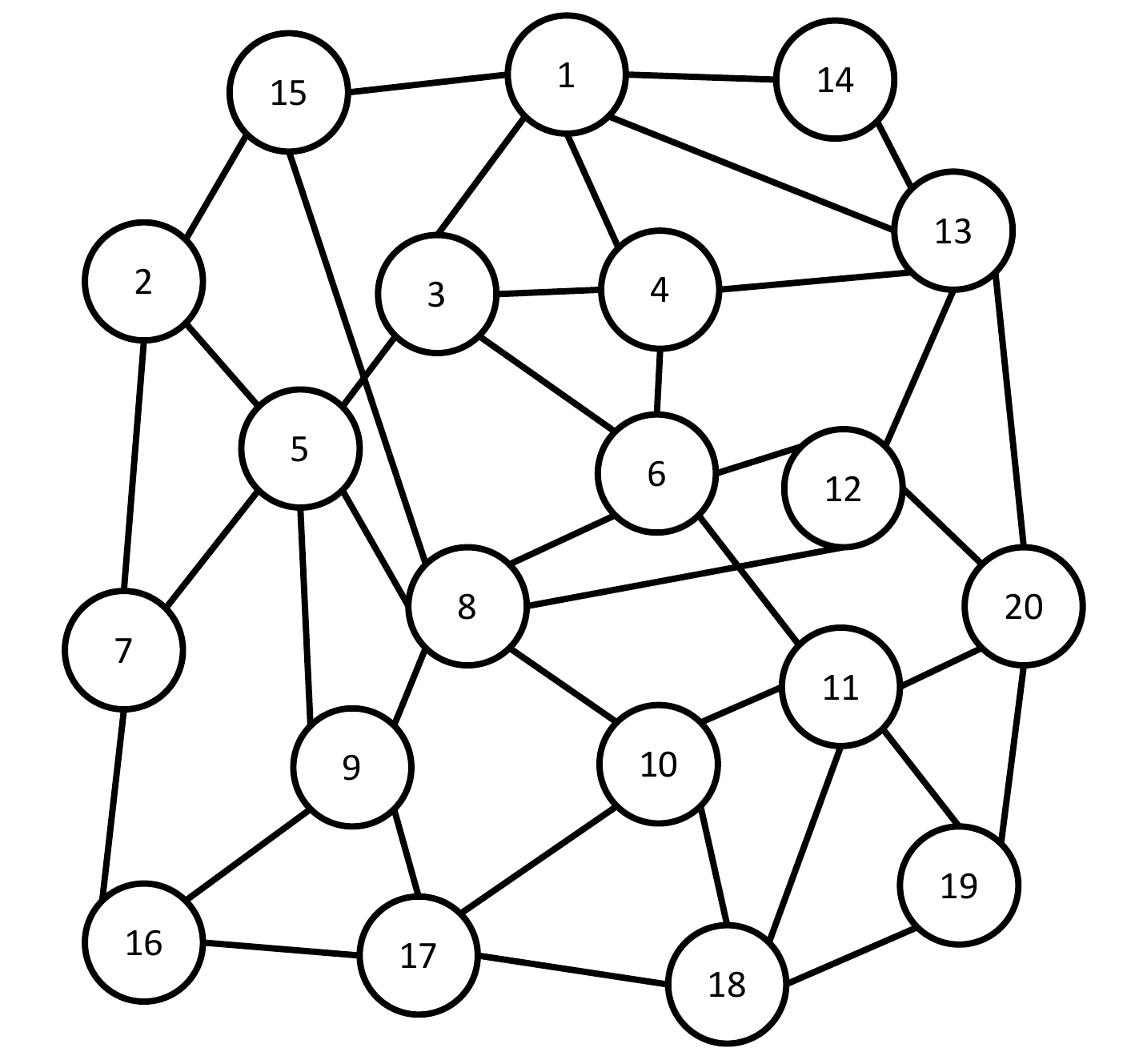}
   \caption{Communication topology of the network.}
    \label{20agents}
  \end{figure}

Figure \ref{SumFxstar_Fx_est_20functions_August6_it17_10} presents on a log-scale the difference between the optimal solution and the estimated solution as a function of the iteration $k$.  We can see that the agent teams indeed quickly reach an approximately optimal solution. 
  
  \begin{figure}[h!]
   \includegraphics[width=0.5\textwidth,height=0.35\textwidth]{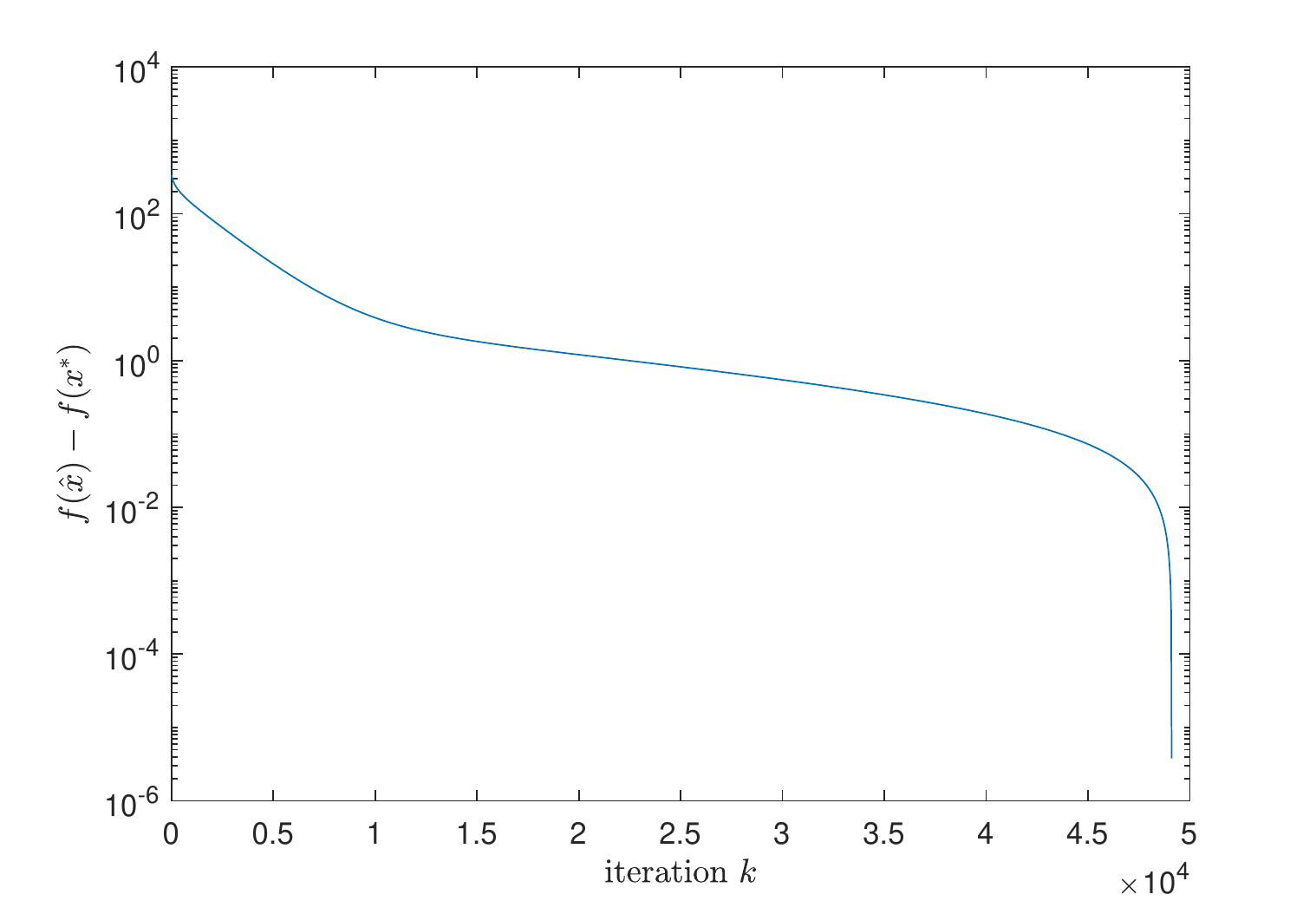}
   \caption{$f(\hat{x})- f(x^*)$  in function of the iteration number $k$ for a network of 20 agents }
    \label{SumFxstar_Fx_est_20functions_August6_it17_10}
  \end{figure}

\section{CONCLUSIONS}
\label{CONCLUSIONS}

This paper proposed a distributed algorithm to optimize a prioritized sum of convex objective functions. The algorithm allows the exploration of the Pareto Front of the overall objectives, which has a direct influence on the agents' final results. We established a rule to determine the weights in agents' state updates. The weight matrix in this paper is row stochastic, which provides a less restrictive communication characterization compared to several related works where the matrix is doubly-stochastic. We studied the convergence of the algorithm and provided explicit bounds for the transition matrix and the agent state update matrix, as well as a performance bound for the algorithm. Future works include implementing the proposed algorithm on a team of robots and investigating real-time changes of the relative importance of the objective functions.







\addtolength{\textheight}{-12cm}   

\end{document}